\documentclass{article}

\usepackage[sglblindworkshop]{neurips_2026}
\workshoptitle{Second Workshop on ML$\times$OR}

\usepackage[utf8]{inputenc}
\usepackage[T1]{fontenc}
\usepackage{hyperref}
\usepackage{url}
\usepackage{booktabs}
\usepackage{amsfonts}
\usepackage{amsmath}
\usepackage{amssymb}
\usepackage{amsthm}
\usepackage{microtype}
\usepackage{graphicx}
\usepackage{enumitem}
\usepackage[capitalize,noabbrev]{cleveref}

\theoremstyle{plain}
\newtheorem{theorem}{Theorem}
\newtheorem{proposition}[theorem]{Proposition}
\newtheorem{corollary}[theorem]{Corollary}
\theoremstyle{definition}
\newtheorem{definition}[theorem]{Definition}
\theoremstyle{remark}
\newtheorem{remark}[theorem]{Remark}

\newcommand{\calC}{\mathcal{C}}
\newcommand{\cov}{\mathrm{cov}}
\newcommand{\Recall}{\mathsf{Recall}}
\newcommand{\GreedyABAC}{\textsc{Greedy-ABAC}}

\title{Generation, Not Selection: LLM Proposals for Certified
Access-Control Policy Mining}

\author{%
  Gautam Kumar\thanks{Indian Institute of Technology Kharagpur,
    \texttt{singhgautam1124@gmail.com}} \qquad
  Ravi Sundaram\thanks{Northeastern University,
    \texttt{r.sundaram@northeastern.edu}} \qquad
  Shamik Sural\thanks{Indian Institute of Technology Kharagpur,
    \texttt{shamik@cse.iitkgp.ac.in}. Corresponding author.}
}

\begin{document}
\maketitle
\nolinenumbers

\begin{abstract}
Attribute-Based Access Control (ABAC) policy mining---recovering a compact,
auditable policy from an access log---is posed as minimum set cover and solved
greedily, which enjoys an $s^{*}(1+\ln m)$ guarantee. That guarantee is
\emph{conditional} on a precondition the literature never states: the candidate
pool handed to greedy must contain the optimal rules. Generation, not
selection, is the binding constraint, and existing generators are purely
\emph{syntactic}---they surface a rule only when enough log records witness it,
so they fail on precisely the rare exception rules that fill a heavy-tailed
policy. We replace the generator with a language model proposing rules
\emph{semantically}, from attribute names, and keep the combinatorial selector
to carry the guarantee: every proposal is filtered against the observed
denials, so soundness is independent of the proposal mechanism and an
unreliable proposer can only waste candidates, never grant access. Over five
seeds on each of two enterprise domains, adding semantic proposal to a
syntactic miner dominates it on \emph{ten of ten} runs---strictly better on
seven---with near-optimal policies ($9.2$--$9.6$ rules against $s^{*}=9$,
versus $11.2$--$14.6$) and false grants cut four- to sevenfold. Replacing the
syntactic miner rather than augmenting it wins only where witnesses are scarce.
An anonymization ablation isolates the mechanism: opaque attribute names cost
two thirds of the model's rule recovery while leaving the syntactic baselines
provably unmoved, and cost nothing where witnesses are plentiful. Semantic
knowledge and log evidence are substitutes.
\end{abstract}

\section{Generation, not selection, is the bottleneck}
\label{sec:decomposition}

Large organizations authorize millions of access requests---who may read which
document, deploy to which environment, approve which trade---against a policy
that must stay small enough for a human to audit. \emph{Policy mining} recovers
such a policy from a log of past decisions, and is the standard route to
migrating a legacy authorization system onto Attribute-Based Access Control
(ABAC)~\citep{XuStoller2015}. It is a combinatorial optimization problem with a
generative subproblem: the objective is a minimum cover, but the items to be
covered must first be \emph{invented}. This paper uses a language model for the
invention and keeps an optimization layer to certify the result---a division of
labour in which the generative component cannot violate the constraint that
matters.

An ABAC rule $r$ is a conjunction of attribute-value predicates and a policy
$\Pi$ a disjunction of rules. Given permitted requests $S^{+}$ ($|S^{+}|=m$)
and denied requests $S^{-}$, \emph{policy mining} seeks a
minimum-cardinality $\Pi$ with $\cov(\Pi,S^{+})=S^{+}$ and
$\cov(\Pi,S^{-})=\emptyset$~\citep{XuStoller2015}: minimum set cover over the
universe $S^{+}$, hence NP-hard~\citep{Karp1972}, and solved in deployed miners
by the greedy heuristic~\citep{XuStoller2015,CotriniBasin2018,AggarwalSural2023}.

Every such pipeline factors into \emph{generation}
($S^{+},S^{-} \to \calC$) followed by \emph{selection}
($\calC \to \hat\Pi$), but is analyzed as one. Greedy searches a finite
\emph{candidate pool} $\calC$, never the full rule space, which is exponential
in the number of attribute-value pairs $A$; the classical $\ln m$ analysis
bounds only the second stage.

\begin{definition}[Optimal rule recall]
\label{def:recall}
For a pool $\calC$ and an optimal policy $\Pi^{*}=\{r_1^{*},\dots,r_{s^{*}}^{*}\}$,
\[
  \Recall_{\tau}(\calC,\Pi^{*}) =
  \frac{1}{s^{*}}\sum_{i=1}^{s^{*}}
  \mathbf{1}\!\left[\max_{r \in \calC}
  \frac{|\cov(r,S^{+}) \cap \cov(r_i^{*},S^{+})|}{|\cov(r_i^{*},S^{+})|}
  \geq \tau\right].
\]
\end{definition}

\begin{theorem}[Conditional greedy guarantee]
\label{thm:conditional}
If $\Recall_{1}(\calC,\Pi^{*})=1$ then \GreedyABAC{} returns $\hat\Pi$ with
$|\hat\Pi| \leq s^{*}(1+\ln m)$, and a sample polynomial in $A$, $s^{*}$ and the
inverse accuracy and confidence parameters suffices for $\hat\Pi$ to be probably
approximately correct on unseen requests. If $\Recall_{1} < 1$ the size bound
does not apply.
\end{theorem}

The size bound is \citet{Chvatal1979} on the universe $S^{+}$, the
generalization bound a consistent-learner union bound. The contribution is the
hypothesis. $\Recall_{1}=1$ is a property of the \emph{generator}, assumed
rather than established in the mining literature, and \Cref{sec:generation}
shows it fails where it matters. It is also measurable whenever ground truth is
known, which makes it our instrument.

\section{Why syntactic generation starves}
\label{sec:generation}

\paragraph{Seed-and-generalize is anchored.}
The dominant generator~\citep{XuStoller2015} is \emph{bottom-up}: for each
$x \in S^{+}$, start from the full conjunction $\mathbf{b}(x)$ and delete
predicates one at a time, accepting a deletion only while negative coverage
stays under $\delta_{\max}$. Every rule it can emit is a subset of a
\emph{single} request's predicates---\emph{anchored} to that request.

\begin{proposition}[Anchoring limits recall]
\label{prop:anchoring}
Let $r^{*} \in \Pi^{*}$ and $x \in \cov(r^{*},S^{+})$. If every rule on every
monotone path from $\mathbf{b}(x)$ to $r^{*}$ in the subset lattice covers
more than $\delta_{\max}|S^{-}|$ denials, then seed-and-generalize applied to
$x$ cannot emit $r^{*}$, nor any generalization of it.
\end{proposition}

The procedure halts at the first dead end, and this is the common case: a request
carries predicates irrelevant to the rule that authorized it, which a deletion
path cannot distinguish from the load-bearing ones.

\paragraph{Group intersection helps, but needs witnesses.}
Intersecting \emph{groups} of requests avoids the trap. Formal concept analysis
(FCA)~\citep{GanterWille1999} takes the context whose objects are permitted
requests and whose attributes are predicates; the closure $\mathrm{cl}(G)$---the
predicates every request in $G$ shares---is the tightest conjunction the whole
group satisfies, hence a candidate rule by construction, with no search.

\begin{proposition}[FCA recovers optimal rules]
\label{prop:fca}
If $r^{*} \in \Pi^{*}$ covers $x_1 \neq x_2$ then
$r^{*} \subseteq \mathrm{cl}(\{x_1,x_2\})$, with equality when $x_1,x_2$ share
no predicate outside $r^{*}$.
\end{proposition}

This is a large improvement, and also the source of a hard limit. $r^{*}$ is
reachable only if some \emph{enumerated} group of size $\leq k$ intersects down
to it, and enumerating all of them costs $O(m^{k}A)$---$5.5\times10^{9}$
triples at $m=3200$, so the enumeration must be subsampled. Both the arity
limit and the sampling budget bite hardest on the rules with fewest
\emph{witnesses}: a rule governing three requests yields three pairs and one
triple, and if those few requests share extraneous predicates, every available
closure over-specifies.

Witness starvation is therefore generic for the tail rather than an artifact,
since real access logs are heavy-tailed---a small core of rules governs most
traffic while a long tail of exceptions (contractor accounts, temporary grants,
legacy overrides) is individually rare~\citep{CotriniBasin2018}. Both syntactic
generators are bounded by how often the log \emph{happens} to exhibit a rule. A
language model is not: it can propose
\texttt{Contractor}~$\wedge$~\texttt{Sandbox} from what those words mean,
whether the log witnesses it three times or three hundred.

\section{Making an unreliable proposer safe}
\label{sec:soundness}

A stochastic proposer inside an authorization pipeline invites the obvious
objection, and the pipeline's structure answers it. Write
$\hat\Pi = \GreedyABAC(S^{+},S^{-},\calC)$ where
$\calC = \{r \in \tilde\calC : \cov(r,S^{-}) = \emptyset\}$ is the
\emph{admitted} pool obtained by validity-filtering the multiset $\tilde\calC$
proposed by an arbitrary mechanism $\mathcal{P}$.

\begin{proposition}[Generator-independent soundness]
\label{prop:soundness}
For every proposal mechanism $\mathcal{P}$---adversarial, stochastic, or
hallucinating---$\cov(\hat\Pi,S^{-}) = \emptyset$.
\end{proposition}
\begin{proof}
$\hat\Pi \subseteq \calC$ and $\cov(\cdot,S^{-})$ is monotone under union, so
$\cov(\hat\Pi,S^{-}) = \bigcup_{r \in \hat\Pi}\cov(r,S^{-}) = \emptyset$.
\end{proof}

\begin{corollary}[One-dimensional risk surface]
\label{cor:risk}
For any $\mathcal{P}$: (i) the false-grant rate on $S^{-}$ is zero
unconditionally; (ii) if $\Recall_1(\calC,\Pi^{*})=1$ then
$|\hat\Pi| \leq s^{*}(1+\ln m)$; (iii) if recall fails, only size and coverage
degrade. The generative component's influence is confined to
$\Recall_{\tau}$, which is measurable.
\end{corollary}

The proposition is elementary, and that is the point: a hallucinated rule wastes
a candidate rather than granting access. Nor is the selector ceremonial---in a
typical run it discards a third of the proposals as unsound and covers $S^{+}$
with roughly half of what survives (Appendix \ref{app:repro}).

\paragraph{Refinement.}
Round one asks for rules cold, and most proposals are rejected without the model
learning why. We close the loop: each rejected rule is returned with its
permitted-versus-denied counts, examples of the denials it matched, and a sample
of still-uncovered permitted requests. One such round is what turns semantic
proposal from competitive into decisive (\Cref{sec:experiments}).

\begin{remark}[What the guarantee does \emph{not} cover]
\label{rem:honest}
Proposition \ref{prop:soundness} is soundness with respect to the \emph{observed} denials.
An overly general rule that avoids every denial in $S^{-}$ may still authorize
unseen requests; controlling \emph{distributional} false grants is the
uniform-convergence half of \Cref{thm:conditional}, not the filter. We therefore
report held-out false-grant rate throughout and treat it, not the filter, as the
evidence on generalization.
\end{remark}

\section{Experiments}
\label{sec:experiments}

\paragraph{Setup.}
Two synthetic enterprise domains with certified ground truth: \textbf{D1}
(general IT, $10$ attribute dimensions, $A=41$) and \textbf{D2} (financial
services, $A=37$), each governed by nine hand-written rules. Ground truth is
required, since $\Recall_{\tau}$ and the excess $|\hat\Pi| - s^{*}$ are not
computable on a real log and no public ABAC dataset ships a certified optimal
policy~\citep{XuStoller2015,CotriniBasin2018}---and the usual benchmark
publishes anonymized integer identifiers, leaving a semantic proposer nothing to
read. Each log has $2{,}000$ permitted and $600$ denied requests, $20\%$ held
out. Rule usage is heavy-tailed ($p_k \propto k^{-2.5}$) and free attributes are
drawn from correlated personas, without which FCA is trivially optimal
(Appendix \ref{app:design}). Four generators---seed-and-generalize (SG), FCA closures
over pairs and triples, LLM, and LLM~$\cup$~FCA---share the \emph{same} validity
filter and selector, so differences are attributable to generation. The proposer
is Claude Opus~5 (max effort) via its web interface, two rounds, five seeds per
domain. See results in Table \ref{tab:main}.

\begin{table}[t]
  \caption{Five seeds per domain, mean $\pm$ sd, $s^{*}=9$; the rarest rule has
  $6$--$12$ witnesses on D1, $19$--$23$ on D2. Held-out false grants are the
  fraction of \emph{unseen} denials authorized by $\hat\Pi$; on \emph{observed}
  denials every pool admits zero by construction (Proposition \ref{prop:soundness}).}
  \label{tab:main}
  \centering
  \small
  \begin{tabular}{@{}llrrrr@{}}
    \toprule
    & & $\Recall_{0.9}$ & $|\hat\Pi|$ & held-out recall & held-out false grants \\
    \midrule
    D1 & FCA            & $0.800 \pm 0.050$ & $14.6 \pm 2.8$ & $0.988 \pm 0.006$ & $0.048 \pm 0.016$ \\
       & LLM            & $0.911 \pm 0.122$ & $\mathbf{8.4 \pm 0.9}$ & $0.994 \pm 0.008$ & $\mathbf{0.003 \pm 0.005}$ \\
       & LLM $\cup$ FCA & $\mathbf{0.956 \pm 0.061}$ & $9.6 \pm 1.3$ & $\mathbf{0.998 \pm 0.003}$ & $0.013 \pm 0.007$ \\
    D2 & FCA            & $0.933 \pm 0.061$ & $11.2 \pm 2.0$ & $0.995 \pm 0.006$ & $0.013 \pm 0.013$ \\
       & LLM            & $0.822 \pm 0.168$ & $\mathbf{8.8 \pm 0.8}$ & $0.975 \pm 0.023$ & $\mathbf{0.003 \pm 0.005}$ \\
       & LLM $\cup$ FCA & $\mathbf{0.978 \pm 0.050}$ & $9.2 \pm 0.4$ & $\mathbf{1.000 \pm 0.000}$ & $\mathbf{0.002 \pm 0.004}$ \\
    \bottomrule
  \end{tabular}
\end{table}

\paragraph{Semantic proposal dominates as a complement, and sometimes alone.}
Adding the LLM pool to FCA's matches or exceeds FCA on \emph{all ten} runs and
strictly exceeds it on seven, and is better on every other axis in both
regimes---smaller policies, higher held-out coverage, false grants $4$ to $7$
times lower. Replacing FCA rather than augmenting it wins only where witnesses
are scarce: the LLM alone beats FCA on five of five D1 seeds and loses on four
of five D2 seeds, a split the mechanism paragraphs below explain.

\paragraph{Refinement is what does it.}
LLM-alone recall is $0.645 \to 0.911$ across the two rounds on D1 and
$0.578 \to 0.822$ on D2, while the unsafe-proposal rate falls from $53\%$ to
$27\%$: told which denials it hit, the model repairs over-generality. On D1 that
is the difference between trailing FCA and beating it.

Every rejection was semantic rather than malformed---no schema violations in
$175$ proposals (Appendix \ref{app:repro})---which is exactly what Proposition
\ref{prop:soundness} absorbs. Seed-and-generalize, the literature baseline, is
omitted from \Cref{tab:main} for space; it recovers $0.53$--$0.58$ of the
optimal rules and inflates the policy four- to fivefold (Appendix \ref{app:perseed}).

\paragraph{Operational cost.}
Exhaustive FCA would enumerate $\binom{m}{3}$ groups---$5.5\times10^{9}$ at
$m=3200$---so it must subsample, whereas the proposer makes two calls on a
prompt flat in $m$ ($20.2$k characters throughout), giving generation cost
$O(1)$ in log size where FCA's is not (Appendix \ref{app:cost}). Latency and per-token
cost we could not measure under the manual protocol, and we flag that as a gap
rather than estimate it.

\paragraph{The mechanism is semantic.}
Replacing attribute names with \texttt{A1}--\texttt{A10} and values with opaque
\texttt{v001}--\texttt{v041}---identical instance, prompt and two rounds---costs
the LLM two thirds of its rule recovery on D1 seed~$0$: $\Recall_{0.9}$ falls
$1.000 \to 0.333$, held-out recall $1.000 \to 0.750$. SG and FCA are
\emph{bit-identical} across conditions on every metric, and structurally must
be: they use only subset relations over predicate bitmasks, renaming is a
bijection on bit positions, and the substitution is applied at the prompt
boundary and inverted before scoring. The baselines are mathematically pinned,
so the LLM is the only component that can move.

Semantics is moreover needed to \emph{use} the feedback and not only to propose:
refinement lifts recall by $0.333$ with names and $0.111$ without. And the
effect vanishes on D2 ($0.78$ either way), where every rule is well witnessed.
Semantics, on this evidence, \emph{substitutes for witnesses}---the same axis
that separates the LLM from FCA---which is why anonymization \emph{inverts} the
D1 ranking, the anonymized proposer ($0.333$) falling below FCA ($0.778$)
(Appendix \ref{app:anon}). On a log of opaque identifiers, use the syntactic miner.

\paragraph{Capability threshold.}
The same prompt, loop and filter give $\Recall_{0.9} = 0.33$ with Claude
Sonnet~5 at medium effort against $1.00$ with Opus~5 at max effort, while
Gemini Flash and Pro are indistinguishable (Appendix \ref{app:tier}). Semantic proposal
is a property of sufficiently capable models rather than of language models in
general, and improves with the proposer at no change to the guarantee. Sweeping
the tail exponent locates the other boundary: FCA is exactly optimal for
$s \leq 1.5$, degrading to $0.72$ at $s = 3.5$ (Appendix \ref{app:zipf}).

\section{Discussion}
\label{sec:discussion}

The decomposition generalizes beyond access control. Wherever a combinatorial
selector operates over a candidate set someone had to generate, its guarantee is
conditional on the generator, and a generative model is a strong generator
precisely where enumeration is expensive or semantically blind. The pairing is
safe because the selector enforces a feasibility constraint the proposer cannot
weaken: the generator proposes, the optimizer certifies.

Limitations. Five seeds on two synthetic domains is a narrow base, and both
replications use a single proposer; LLM-alone superiority holds only in the
witness-scarce regime. Provenance is weaker than an API run: queries went
through web interfaces in August 2026, where consumer tiers do not pin a model
version, so exact re-execution is impossible and we instead release the harness,
generators, prompts and verbatim replies, from which every number is
recomputable (Appendix \ref{app:repro}). Remark \ref{rem:honest} applies throughout---the
unconditional guarantee is soundness on observed denials, not distributional.
Finally, the anonymization result should be read as a \emph{measured utility
cost} of attribute-name anonymization, on a fixed instance with provably unmoved
baselines; whether that cost outweighs the privacy benefit it purchases is not
something this experiment settles.

\bibliographystyle{plainnat}
\bibliography{references}

\appendix
\section{Benchmark design}
\label{app:design}
Free attributes are drawn from correlated personas rather than independently,
and this is essential rather than decorative. Under independent sampling, any
two positives covered by the same optimal rule share \emph{only} that rule's
predicates, so Proposition \ref{prop:fca} holds with equality almost surely and FCA
recovers every rule from almost any pair---scoring a degenerate
$\Recall_{1.0}=1.0$ on which no generator can be distinguished from another. Our
first generator had exactly this defect. Correlation is what makes a pairwise
closure \emph{over-specify}: the two requests share extraneous predicates beyond
the rule that authorized them, which is also what real organizations look like,
engineers skewing to one region, to managed devices, to source repositories.

Rule usage follows $p_k \propto k^{-s}$ with $s=2.5$, leaving D1's rarest rule
$6$--$12$ witnesses and D2's $19$--$23$. The two domains were written before any
model was run, and D2 was not selected for being easier; the witness gap between
them is a consequence of the schema sizes, and is what makes the pair useful.

\section{Per-seed results, both domains}
\label{app:perseed}
\Cref{tab:main} reports means; per-seed values are below, including seed-and-generalize (SG), which the main table omits for space. $s^{*}=9$ throughout.

The residual errors are benign. On D1 seed~$0$ the LLM recovers eight of the nine ground-truth rules bit-exactly, and misses the ninth only because the log cannot distinguish it from a strictly more general rule covering the same $54$ positives---which leaks the single held-out denial that is that seed's entire false-grant rate. The failure mode is an unidentifiable rule rather than a wrong one.

\begin{table}[h]
  \caption{D1, general IT (rarest rule 6--12 witnesses). Per seed: $\Recall_{0.9}$ / $|\hat\Pi|$ / held-out false-grant rate.}
  \centering\small
  \begin{tabular}{@{}lccccc@{}}
    \toprule
    & seed 0 & seed 1 & seed 2 & seed 3 & seed 4 \\
    \midrule
    SG & 0.22 / 55 / 0.150 & 0.67 / 18 / 0.050 & 0.56 / 61 / 0.108 & 0.56 / 36 / 0.100 & 0.67 / 37 / 0.058 \\
    FCA & 0.78 / 14 / 0.050 & 0.78 / 12 / 0.042 & 0.89 / 12 / 0.042 & 0.78 / 17 / 0.033 & 0.78 / 18 / 0.075 \\
    LLM & 1.00 / 9 / 0.008 & 0.78 / 8 / 0.008 & 1.00 / 9 / 0.000 & 0.78 / 7 / 0.000 & 1.00 / 9 / 0.000 \\
    LLM $\cup$ FCA & 1.00 / 9 / 0.008 & 0.89 / 9 / 0.017 & 1.00 / 9 / 0.008 & 0.89 / 12 / 0.025 & 1.00 / 9 / 0.008 \\
    \bottomrule
  \end{tabular}
\end{table}

\begin{table}[h]
  \caption{D2, financial services (rarest rule 19--23 witnesses). Per seed: $\Recall_{0.9}$ / $|\hat\Pi|$ / held-out false-grant rate.}
  \centering\small
  \begin{tabular}{@{}lccccc@{}}
    \toprule
    & seed 0 & seed 1 & seed 2 & seed 3 & seed 4 \\
    \midrule
    SG & 0.44 / 73 / 0.175 & 0.89 / 21 / 0.050 & 0.44 / 46 / 0.092 & 0.44 / 46 / 0.125 & 0.67 / 37 / 0.108 \\
    FCA & 1.00 / 9 / 0.000 & 0.89 / 13 / 0.008 & 0.89 / 12 / 0.033 & 1.00 / 9 / 0.008 & 0.89 / 13 / 0.017 \\
    LLM & 0.78 / 8 / 0.000 & 1.00 / 9 / 0.000 & 0.89 / 8 / 0.000 & 0.89 / 9 / 0.008 & 0.56 / 10 / 0.008 \\
    LLM $\cup$ FCA & 1.00 / 9 / 0.000 & 1.00 / 9 / 0.000 & 1.00 / 9 / 0.000 & 1.00 / 9 / 0.008 & 0.89 / 10 / 0.000 \\
    \bottomrule
  \end{tabular}
\end{table}

\section{Proposer capability}
\label{app:tier}
All five configurations received the identical prompt, refinement loop and validity filter; only the proposer differs. Round-1 figures isolate raw proposal quality.

\begin{table}[h]
  \caption{$\Recall_{0.9}$ of the LLM-only pool, seed~0, round~1. FCA scores $0.78$ (D1) and $1.00$ (D2) on the same instances.}
  \centering\small
  \begin{tabular}{@{}lcc@{}}
    \toprule
    proposer & D1 & D2 \\
    \midrule
    ChatGPT (GPT-5.6, free tier) & $0.22$ & $0.33$ \\
    Claude Sonnet 5, medium effort & $0.11$ & $0.56$ \\
    Claude Opus 5, max effort & $\mathbf{0.67}$ & $0.56$ \\
    Gemini Flash, regular thinking & $0.44$ & $\mathbf{0.78}$ \\
    Gemini Pro, extended thinking & $0.44$ & $0.67$ \\
    \bottomrule
  \end{tabular}
\end{table}

Tier dominates model family. Sonnet~5 medium to Opus~5 max is $0.11 \to 0.67$ on D1 and, after two rounds, $0.33 \to 1.00$; Gemini Flash to Pro is $0.44 \to 0.44$ on D1 and $0.78 \to 0.67$ on D2---no gain, then a loss. An early pilot in which a small model appeared to beat a large one was therefore a tier artifact rather than a model-family finding.

\section{Where the syntactic baseline breaks: tail exponent}
\label{app:zipf}
Rule usage follows $p_k \propto k^{-s}$; larger $s$ starves the rarest rules of witnesses. Three seeds, both domains, FCA only, since the boundary is a property of the baseline.

\begin{table}[h]
  \caption{FCA degrades monotonically once the tail is heavy enough. `min witnesses' is the rarest ground-truth rule's support.}
  \centering\small
  \begin{tabular}{@{}lccccccc@{}}
    \toprule
    $s$ & $0$ & $1$ & $1.5$ & $2$ & $2.5$ & $3$ & $3.5$ \\
    \midrule
    min witnesses & $66$ & $76$ & $49$ & $30$ & $15$ & $10$ & $5$ \\
    FCA $\Recall_{0.9}$ & $1.00$ & $1.00$ & $1.00$ & $0.96$ & $0.87$ & $0.82$ & $0.72$ \\
    FCA $|\hat\Pi|$ & $9.0$ & $9.0$ & $9.0$ & $11.2$ & $12.0$ & $13.2$ & $14.5$ \\
    \bottomrule
  \end{tabular}
\end{table}

The main experiments use $s=2.5$. FCA is \emph{exactly} optimal for $s \leq 1.5$: on mildly-tailed policies a syntactic generator needs no help, and a semantic proposer earns its cost only past the point where the tail leaves rules with single-digit support.

\section{Generation cost against log size}
\label{app:cost}
\begin{table}[h]
  \caption{Two seeds, $s=2.5$. Exhaustive FCA is infeasible, so the enumeration is capped at $2\times10^{5}$ sampled groups; the price of the cap is coverage of the closure space, not wall time. The prompt is flat in $m$ by construction, the proposer seeing a fixed sample of the log.}
  \centering\small
  \begin{tabular}{@{}lccccc@{}}
    \toprule
    $m$ & $200$ & $400$ & $800$ & $1600$ & $3200$ \\
    \midrule
    exhaustive $\binom{m}{3}$ & $1.3\times10^{6}$ & $1.1\times10^{7}$ & $8.5\times10^{7}$ & $6.8\times10^{8}$ & $5.5\times10^{9}$ \\
    FCA time (s, capped) & $0.56$ & $0.70$ & $1.07$ & $1.07$ & $1.55$ \\
    SG time (s) & $0.02$ & $0.04$ & $0.07$ & $0.14$ & $0.36$ \\
    LLM prompt (chars) & $20{,}261$ & $20{,}247$ & $20{,}250$ & $20{,}212$ & $20{,}154$ \\
    \bottomrule
  \end{tabular}
\end{table}

Latency and per-token cost are not reported: the manual protocol of Appendix \ref{app:repro} makes them unmeasurable, and we prefer a stated gap to an estimate.

\section{Anonymization ablation in detail}
\label{app:anon}
Attribute names become \texttt{A1}--\texttt{A10} and values opaque, shuffled \texttt{v001}--\texttt{v041}. The substitution is applied only when rendering the prompt and inverted before scoring, so SG and FCA run over the untouched instance through identical code. Their invariance is structural rather than observed---both use only subset relations over predicate bitmasks, and renaming is a bijection on bit positions---and we verified it numerically, all six metrics identical in both conditions.

\begin{table}[h]
  \caption{Claude Opus~5 max effort, seed~0, two rounds in each condition.}
  \centering\small
  \begin{tabular}{@{}lcccc@{}}
    \toprule
    & \multicolumn{2}{c}{D1 (witnesses scarce)} & \multicolumn{2}{c}{D2 (plentiful)} \\
    & named & anonymized & named & anonymized \\
    \midrule
    LLM $\Recall_{0.9}$ & $1.000$ & $\mathbf{0.333}$ & $0.778$ & $0.778$ \\
    LLM held-out recall & $1.000$ & $0.750$ & $0.968$ & $0.968$ \\
    unsafe proposals, round 2 & --- & $28/30$ & --- & --- \\
    refinement lift, recall & $+0.333$ & $+0.111$ & --- & --- \\
    \midrule
    SG, FCA (all metrics) & \multicolumn{4}{c}{identical across conditions} \\
    \bottomrule
  \end{tabular}
\end{table}

Two readings follow. Semantics is needed to \emph{use} the refinement feedback and not only to propose: with names one round adds $0.333$ recall, without them $0.111$, and $28$ of $30$ round-two proposals covered an observed denial. And anonymization costs nothing on D2, where every rule is well witnessed---semantic knowledge and log evidence substitute for one another. Note also that the anonymized proposer scores $0.333$ against FCA's $0.778$ on D1: anonymization does not merely degrade the proposer, it reverses which method a practitioner should choose. We read this as a measured utility cost of attribute-name anonymization, not as an argument against it.

\section{Reproducibility and provenance}
\label{app:repro}
Model queries were issued through consumer web interfaces on 2026-08-13 (ChatGPT GPT-5.6 free tier; Claude Sonnet~5 medium effort; Claude Opus~5 max effort; Gemini Flash regular thinking; Gemini Pro extended thinking), API access being unavailable to us. Prompts are produced by the same function an API path would call, so the two are comparable, and every reply is stored verbatim. Replies are parsed leniently---chat interfaces add code fences and commentary---but validated strictly: any predicate outside the schema, or two values for one attribute, is discarded and counted. Consumer tiers do not pin a model version, so exact re-execution is impossible; every reported number is instead recomputable from the recorded replies. We release the harness, the domain generators, the prompt templates, and the replies at \url{https://github.com/koods1/semantic-abac-mining}, archived at \url{https://doi.org/10.5281/zenodo.22167397}.

To make concrete how much work the selector does, take D1 seed~$0$: the proposer offered $22$ rules across the two rounds, the validity filter discarded the $7$ that granted an observed denial, and greedy covered $S^{+}$ with $9$ of the surviving $15$. Emitting the raw model output as the policy would have shipped $22$ rules, $7$ of them unsound. Aggregated over all $175$ proposals from the five configurations, $53\%$ were unsound at round~1 and $27\%$ at round~2, and \emph{none} was malformed.

\section{Acknowledgments}
\label{ack}
We acknowledge the use of Claude to edit this document.

\end{document}